\documentclass[runningheads,envcountsame]{llncs}
\usepackage[T1]{fontenc}

\usepackage{amsthm}
\usepackage{amssymb}
\usepackage{thmtools,thm-restate}

\usepackage{tikz}
\usepackage{xspace}

\usepackage{booktabs}   %% For formal tables:
\usepackage{subcaption} %% For complex figures with subfigures/subcaptions
\usepackage{enumitem} 
\usepackage{makecell} 
\usepackage{mathtools}
\usepackage{wrapfig}

\usepackage{amsmath}
\usepackage{pifont}
\usepackage{thm-restate}
\usetikzlibrary{arrows,automata,shapes,decorations,decorations.markings,calc, matrix,decorations.pathmorphing, patterns,backgrounds,shapes.misc,arrows.meta,positioning}

\usepackage{xspace}

\usepackage[ruled,vlined,resetcount,linesnumbered,noend]{algorithm2e}

\usepackage{parskip}

\usepackage{comment}

\usepackage{arydshln}
 \usepackage{multicol}
 \usepackage{multirow}

\usepackage{hyperref}
\usepackage[capitalise]{cleveref}

\usepackage{varwidth}

\usepackage{array}
\newcolumntype{H}{>{\setbox0=\hbox\bgroup}c<{\egroup}@{}}

\usepackage{ifthen}
\usepackage{stackengine}

\usepackage[normalem]{ulem}

\usepackage{todonotes}
  \newcounter{todocounter}

\newcommand\varset{{\mathbb X}}

\newcommand\valset{{\mathbb V}}

\newcommand\lbl\ell
\newcommand\action\lbl

\newcommand\run\rho
\newcommand\thread{S} 
\newcommand\program{\mathsf{Prog}}
\newcommand{\inter}{\rho}

\newcommand{\pp}{\pi} 
\newcommand{\vscp}{\textsf{VSC}$^\pp$}
\newcommand{\points}{\mathcal{P}}
\newcommand{\opset}{\Sigma}
\newcommand{\eventset}{\Sigma_{\program}}

\colorlet{colorPO}{darkgray!80!black}
\colorlet{colorRF}{blue}
\colorlet{colorOB}{orange}
\colorlet{colorCO}{red!80!black}
\colorlet{colorMO}{red!80!black}
\colorlet{colorFR}{olive}
\colorlet{colorECO}{orange}
\colorlet{colorCOM}{magenta}
\colorlet{colorSW}{teal}
\colorlet{colorHB}{green!40!black}
\colorlet{colorPPO}{magenta}
\colorlet{colorRSEQ}{green!40!black}
\colorlet{colorSC}{violet}
\colorlet{colorHBMO}{brown}
\colorlet{colorPSC}{violet}
\colorlet{colorREL}{olive}
\colorlet{colorWB}{olive}
\colorlet{colorRMW}{brown}
\newcommand\prog{{\mathcal P}}

\newcommand\conf\gamma

\tikzset{
   every path/.style={>=stealth},
   po/.style={->,draw=colorPO,shorten >=-0.5mm,shorten <=-0.5mm},
   ppo/.style={->,draw=colorPPO,shorten >=-0.5mm,shorten <=-0.5mm},
   sw/.style={->,draw=colorSW,shorten >=-0.5mm,shorten <=-0.5mm},
   rf/.style={->,draw=colorRF,dashed,shorten >=-0.5mm,shorten <=-0.5mm},
   mo/.style={->, draw=colorMO,dotted,shorten >=-0.5mm,shorten <=-0.5mm},
    mrf/.style={->,draw=colorRF,dashed,shorten >=-0.5mm,shorten <=-0.5mm},   
   urf/.style={->,draw=colorRF,shorten >=-0.5mm,shorten <=-0.5mm},
   fr/.style={->,draw=colorFR,dashed,shorten >=-0.5mm,shorten <=-0.5mm},
   hb/.style={->,draw=colorHB,thick,shorten >=-0.5mm,shorten <=-0.5mm},
   co/.style={->,draw=colorCO,dotted,very thick,shorten >=-0.5mm,shorten <=-0.5mm},
   rmw/.style={->,draw=colorRMW,thick,shorten >=-0.5mm,shorten <=-0.5mm},
   rseq/.style={->,draw=colorRSEQ,thick,dotted,shorten >=-0.5mm,shorten <=-0.5mm},
   com/.style={->,draw=colorCOM,thick,shorten >=-0.5mm,shorten <=-0.5mm},
}

\renewcommand{\emptyset}{\varnothing}
\newcommand{\nats}{\mathbb{N}}

\newcommand{\scmm}{\mathsf{SC}}

\newcommand{\po}{\mathsf{\color{colorPO}po}}

\newcommand{\rd}{\mathtt{r}}
\newcommand{\wt}{\mathtt{w}}

\newcommand{\NP}{\mathsf{NP}}

\newcommand{\Sel}{\mathsf{Sel}}
\newcommand{\Checker}{\mathsf{Checker}}
\newcommand{\BB}{\mathsf{B}}
\newcommand{\Start}{\mathsf{Start}}
\newcommand{\Finish}{\mathsf{Finish}}

\SetKwProg{myfun}{function}{}{}
\SetKwProg{myhandler}{handler}{}{}
\SetKwFunction{init}{Initialization}
\SetKwFunction{getLW}{getLastWriteBeforeOffline}
\SetKwFunction{poprop}{poPropagate}
\SetKwFunction{joinchangedindex}{pairwiseMaxIndexes}
\SetKwFunction{getLWB}{getLastWriteBefore}
\SetKwFunction{getLRB}{getLastReadBefore}
\SetKwFunction{getLRWB}{getLastReadWriteBefore}
\SetKwFunction{initWtLst}{initWtLst}
\SetKwFunction{initRdLst}{initRdLst}
\SetKwFunction{initWtRdLst}{initWtRdLst}
\SetKwFunction{get}{get}
\SetKwFunction{getHBTS}{getHBTimestamps}
\SetKwFunction{getTCPC}{getTopAndPositionInRFChain}
\SetKwFunction{checkRace}{checkRace}
\SetKwFunction{rdhandler}{readHandler}
\SetKwFunction{wthandler}{writeHandler}
\SetKwFunction{acqhandler}{acquire}
\SetKwFunction{relhandler}{release}
\SetKwInput{Input}{Input}
\SetKwInOut{Output}{Output}
\SetKw{Let}{let}
\SetKw{Break}{break}
\SetKw{NOT}{not}
\SetKw{declare}{declare}
\SetKw{exit}{exit}
\SetKw{Continue}{continue}
\SetKw{Return}{return}
\SetKwFor{Foreach}{for each}{}{}%
\DontPrintSemicolon

\SetCommentSty{mycommfont}
\SetNoFillComment

\tikzset{
   every path/.style={>=stealth},
   po/.style={->,color=colorPO, thick},
   ppo/.style={->,color=colorPPO, thick},
   sw/.style={->,color=colorSW, thick},
   rf/.style={->,color=colorRF,dashed, thick},
   mrf/.style={->,color=colorRF,dashed, thick},   
   urf/.style={->,color=colorRF, thick},
   fr/.style={->,color=colorFR,dashed, thick},
   hb/.style={->,color=colorHB,thick, thick},
   mo/.style={->,color=colorMO,dotted, very thick},
   dob/.style={->,color=colorDOB, very thick},
   ob/.style={->,color=colorOB, very thick},
   rmw/.style={->,color=colorRMW,thick, thick},
   rseq/.style={->,color=colorRSEQ,thick,dotted, thick},
   com/.style={->,color=colorCOM,thick, thick},
}

\makeatletter
\newcommand{\labitem}[2]{%
\def\@itemlabel{\textbf{#1}}
\item
\def\@currentlabel{#1}\label{#2}}
\makeatother

\newcommand{\Init}{\mathsf{Init}}

\SetKwFunction{ChCon}{CheckConsistency}
\SetKwFunction{ChPCon}{CheckPreemptionBoundedConsistency}
\SetKwComment{Comment}{/* }{ */}
\SetKwFunction{Explore}{Explore}
\SetKwFunction{ExploreSC}{ExploreSC}
\SetKwFunction{enable}{EnableEvent}
\SetKwFunction{enableth}{EnableThread}
\SetKwFunction{enablethsc}{EnableThreadSC}
\SetKwFunction{buildgraph}{BuildGraph}
\SetKwProg{Fn}{Function}{:}{}

\newcommand{\threewriter}{\textsc{3-Writer}}
\newcommand{\twowriter}{\textsc{2-Writer}}
\newcommand{\onewriter}{\textsc{1-Writer}}

\title{Verifying Sequential Consistency under Bounded Preemptions}
\titlerunning{Verifying Sequential Consistency under Bounded Preemptions}
\author{
R.~Govind\inst{1,2}
\and
S.~Krishna\inst{3}
\and
Sanchari Sil\inst{4}
\and
B.~Srivathsan\inst{4}
}
\institute{
  The Institute of Mathematical Sciences, Chennai, India,
  \email{govind@imsc.res.in}\\
  \and
  Homi Bhabha National Institute, Mumbai, India
  \and
  IIT Bombay, India, 
  \email{krishnas@ cse.iitb.ac.in}\\
  \and
  Chennai Mathematical Institute, India,
  \email{\{sanchari, sri\} @ cmi.ac.in}
}
\authorrunning{R.~Govind et al.}

\begin{document}

\maketitle

\begin{abstract}
    Gibbons and Korach studied a fundamental problem in 1997: given an observed sequence of reads and writes of a multi-threaded program, does there exist \emph{an interleaving} which is sequentially consistent? Apart from applications in testing shared memory implementations, a procedure for this problem is employed in Dynamic Partial-Order-Reduction (DPOR) algorithms. The problem is known to be $\NP$-hard even when different syntactic parameters are kept bounded. 
    
    In this paper, we consider a restriction on the kind of interleaving required: does there exist a sequentially-consistent interleaving with at most $\pp$ preemptions?  
    Empirical evidence suggests that several bugs manifest within a few preemptive switches. This motivates us to investigate the problem under bounded preemptions. Our results exhibit a trichotomy: the problem lends to a polynomial-time algorithm for the class of single-writer programs where for each variable, there is a single thread writing to it; it becomes NP-hard for two-writer programs and finally, for three-writer programs, we get a conditional lower bound under the Exponential-Time-Hypothesis. When the number of preemptions $\pp$ is not bounded, we show the problem to be W[1]-hard, and hence unlikely to be fixed-parameter-tractable with parameter $\pp$. 
\end{abstract}

\section{Introduction}

In~\cite{Gibbons}, Gibbons and Korach considered the following problem.
\begin{center}
  \fbox{%
    \begin{minipage}{0.98\textwidth}
      \noindent \textsc{Verifying Sequential Consistency} (The VSC-problem~\cite{Gibbons})

      \medskip

      \noindent \textbf{Instance:} A variable set $\varset$, a value set $\valset$, a finite collection of nonempty sequences $S_1, S_2, \dots, S_k$, each consisting of a finite set of memory operations of the form ``$read(x,d)$'' or ``$write(x,d)$'' with $x \in \varset$ and $d \in \valset$. 

      \medskip

      \noindent \textbf{Question:} Is there a sequence $S$, an interleaving of $S_1, S_2, \dots, S_k$ such that for each ``$read(x, d)$'' there is a preceding ``$write(x,d)$'' in $S$ with no other ``$write(x, d')$'' in between.
    \end{minipage}%
  }
\end{center}
This problem has applications in \emph{testing shared memories}~\cite{Gibbons} and in dynamic partial order reduction (DPOR) algorithms for verifying multi-threaded programs~\cite{rv-equivalence}. In the context of testing shared memories, the goal is to test whether an implementation of a shared memory system is faithful to sequential consistency. To do that, the test observes a sequence of reads and writes on the memory locations at each thread of the program and determines whether the observed set of sequences admits a sequentially consistent interleaving. This is indeed the exact formulation of the VSC-problem. In the context of DPOR algorithms, the VSC-problem appears in the procedure for stateless model-checking under the \emph{reads-value-from} equivalence~\cite{rv-equivalence}. The reads-value-from equivalence abstracts the executions of a multi-threaded program in terms of the values read by the memory locations. The question of whether an abstract execution admits a concrete execution is simply a reformulation of the VSC-problem.  Clearly, the simple-to-state VSC-problem has critical applications in the verification and analysis of multi-threaded programs. 

The VSC-problem is known to be notoriously hard. It was shown to be $\NP$-hard even when the input instance has only $3$ threads. The complexity arises due to the fact that there are exponentially many possible interleavings and an exhaustive enumeration of all of them seems unavoidable. This raises a natural question of investigating restrictions to the problem that admit a polynomial-time solution. One may consider bounding various aspects of the program syntax, like the number of threads and the number of variables. The problem continues to remain $\NP$-hard under these usual restrictions (see Table~$1$ of \cite{Gibbons}).

In this work, we consider restricting an aspect of the interleaving required. We look at interleavings where there are a bounded number of \emph{preemptions}. A preemption occurs when a currently executing thread is stopped, and the execution moves on to a different thread, and comes back to it later. Empirical evidence suggests that several bugs in multithreaded programs can be caught with very few preemptions~\cite{QadeerMusuvati}, which has resulted in the study of DPOR methods geared towards finding interleavings with bounded preemptions~\cite{Viktor-Buster-23}. This has motivated us to study the VSC-problem under bounded preemptions.

\begin{center}
  \fbox{%
    \begin{minipage}{0.98\textwidth}
      \noindent \textsc{\textbf{VSC$^\pp$-problem}: Verifying Sequential Consistency under $\pp$ preemptions} 

      \medskip

      \noindent \textbf{Fixed:} A constant $\pp \in \mathbb{N}\cup\{ 0\}$

      \medskip 

      \noindent \textbf{Instance:} A variable set $\varset$, a value set $\valset$, a finite collection of nonempty sequences $S_1, S_2, \dots, S_k$, each consisting of a finite set of memory operations of the form ``$read(x,d)$'' or ``$write(x,d)$'' with $x \in \varset$ and $d \in \valset$. 

      \medskip

      \noindent \textbf{Question:} Is there a sequence $S$, an interleaving of $S_1, S_2, \dots, S_k$ \textcolor{blue}{with at most $\pp$ preemptions} such that for each ``$read(x, d)$'' there is a preceding ``$write(x,d)$'' in $S$ with no other ``$write(x, d')$'' in between.
    \end{minipage}%
  }
\end{center}

\subsection{Our Results}

We divide our instances to the \vscp\ problem broadly into two categories: (1) instances where there is a single thread writing to each variable, whereas multiple threads can read on the same variable (called \onewriter\ programs) and (2) instances where multiple threads can write to the same variable. The original VSC-problem is known to be $\NP$-hard even for \onewriter\ programs. Our first key result is the following. 
 
\begin{restatable}{theorem}{sconewriter}\label{th:onewriter} \vscp  admits a polynomial-time solution for \onewriter\ programs.  
\end{restatable}

Our second result shows that \vscp\ turns out to be $\NP$-hard when we allow multiple writers per variable, in fact even with two threads writing to a variable (called \twowriter\ programs).

\begin{restatable}{theorem}{scNPHardness}\label{th:NPhardness}
\vscp is $\NP$-complete, for all $\pi \ge 0$. Hardness holds even for \twowriter\ programs. 
\end{restatable}
 
Thirdly, we prove a conditional lower bound under the Exponential Time Hypothesis (ETH). ETH states that there is no $2^{o(k)} \cdot (k + m)^{\mathcal{O}(1)}$ algorithm for 3-CNF satisfiability, with $k$ being the number of variables and $m$ the number of clauses.  

\begin{restatable}{theorem} {scFinegrained}\label{th:Finegrained}
Under ETH, there is no $2^{o(k)} \cdot n^{\mathcal{O}(1)}$ algorithm for \vscp where $k$ is the number of threads, $n$ is the total number of operations.

The result holds even when the input instance is a \threewriter\ program. 
\end{restatable}

Finally, we investigate the complexity when $\pp$ is not fixed and treated as part of the input. 

\begin{restatable}{theorem}{wonehardness}\label{thm:w1-hard}
When $\pp$ is not fixed, and is treated as a parameter in the input, the \vscp\ problem is W[1]-hard.
\end{restatable}

W[1] is a complexity class studied in the literature on parameterized complexity~\cite{DowneyFellows1999}. Intuitively, W[1]-hardness means we cannot expect an algorithm whose running time is of the form $f(\pp) \cdot (n + k)^{\mathcal{O}(1)}$. 

After presenting some preliminaries in Section~\ref{sec:prelims}, we present the polynomial-time solution for \onewriter\ programs in Section~\ref{sec:one-writer} and the hardness results in Section~\ref{sec:hardness}. 

\subsection{Related work}

\onewriter\ programs were considered in~\cite{rv-equivalence}, with the motivations stated as producer-consumer systems or instances from concurrent-read-exclusive-write (CREW) model. For the original VSC-problem,~\cite{rv-equivalence} presents an $\mathcal{O}(n^{k+1})$ algorithm. Hence when the number of threads is bounded, the VSC-problem for \onewriter\ programs is polynomial-time. In this work, we do not bound the number of threads. Counterparts of the VSC-problem have been studied in the setting of weak memory models~\cite{ChakrabortyKMP24}. Recently, we have studied the problem under variants of the release-acquire semantics~\cite{fm,bm-arxiv}. We obtained a similar trichotomy for the release-acquire semantics. Interestingly, for the sequential consistency memory model that we consider in this work, we obtain the same kind of trichotomy if we incorporate bounded preemptions. For the release-acquire semantics, the solution requires synthesizing what are called reads-from relations and modification-orders. Here, in SC, we need an interleaving. Therefore, even though the results paint the same picture, the techniques involved to obtain the results are completely different. 

\section{Preliminaries}
\label{sec:prelims}

We will fix a finite set $\varset$ of variables (also called memory locations) and a finite value set $\valset$ representing the domain of the variables. There are two kinds of \emph{memory operations}: reading value $d$ on variable $x$, denoted succinctly as $\rd(x, d)$, and writing value $d$ on variable $x$, denoted as $\wt(x, d)$. Let $\opset$ denote the set of memory operations. Since $\varset$ and $\valset$ are finite, $\Sigma$ is finite, and can be seen as a finite alphabet. We make use of standard notations $\Sigma^*$ (resp. $\Sigma^+$) to denote the set of words (resp. non-empty words) from the alphabet $\Sigma$. 

\noindent{\bf{Threads}} A \emph{thread} $\thread$ is a finite (non-empty) sequence of memory operations. In order to distinguish between memory operations appearing in different threads, we equip threads with a unique identifier. Throughout, we assume there are $k$ threads, for some $k \in \nats$, and each thread has an identifier in $\{1, 2, \dots, k\}$. Hence, a thread with identifier $i$ can be seen as a word in $(\{i\} \times \Sigma)^+$. An example of a thread with identifier $2$ would be: 
\begin{align*}
(2: \wt(x, 4)) (2: \rd(y, 3)) (2: \rd(x, 1))
\end{align*}
where $x, y \in \varset$ and $4, 3, 1 \in \valset$. For convenience of exposition, we will call a pair consisting of an identifier and a memory operation as an \emph{event}. Therefore, each thread is a sequence of events with each of them having the same identifier. Since it is a sequence, there is a natural \emph{program order} associated with the events appearing in a thread.

\smallskip 

\noindent{\bf{Concurrent Programs}} A \emph{concurrent program} $\program$ is given by set of threads $S_1, \dots, S_k$, with $S_i$ having identifier $i$. Let $\eventset$ be the set of all events appearing in $\prog$.  There is a natural \emph{program order} $\po$ induced on $\eventset$: for $e_1, e_2 \in \eventset$, we say $e_1~\po~e_2$ if $e_1, e_2$ belong to the same thread, and $e_1$ immediately precedes $e_2$ in the thread sequence. An \emph{interleaving} $\inter$ of $\program$ is an enumeration of $\eventset$ such that the program order is maintained. A \emph{partial interleaving} is an interleaving of some subset of $\eventset$. Since an interleaving is also a (partial) interleaving,  we present our definitions on partial interleavings.

\smallskip 

\noindent{\bf{Conflicting Pair of Read and Write}} A write $\wt(x,d)$ and a read $\rd(x,d')$ are said to be in conflict when they are on the same variable $x$ 
and $d \neq d'$.

\noindent{\bf{Sequential Consistency}} A partial interleaving $\sigma$ is \emph{sequentially consistent} (SC) if for every $\rd(x, d)$ operation appearing in $\sigma$, there is a preceding $\wt(x, d)$ in $\sigma$ with no other operation $\wt(x, d')$ ($d\neq d'$) in between. A (partial) SC interleaving is one which is sequentially consistent.

\noindent{\bf{Preemptions}} Fix a concurrent program $\prog = (S_1, \dots, S_k)$ for the rest of this section. Let $\sigma = a_1 a_2 \dots a_m$, with each $a_j \in \eventset$, be a partial interleaving of $\prog$. We say that a \emph{preemption} occurs at position $j \in \{1, \dots, m-1\}$ if $a_j$ and $a_{j+1}$ belong to different threads, and $a_j$ is not the last event of its thread. For $\pp \in \nats$, we say that $\sigma$ has $\pp$ preemptions if there are exactly $\pi$ positions in $\sigma$ where preemptions occur. Figure~\ref{fig-onewriter-example} shows some examples.

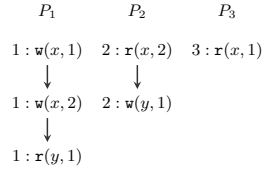
\begin{figure}[h]
   \def\ystep{0.5}

\centering
	\scalebox{.7}{
\begin{tikzpicture}[yscale=1]
\node (t01) at (-.2, 1.5*\ystep) {$P_1$};
\node (t02) at (1.5, 1.5*\ystep) {$P_2$};
\node (t03) at (3.2, 1.5*\ystep) {$P_3$};

  \node (t11) at (-.2, 0*\ystep) {$1:\wt(x,1)$};
    \node (t12) at (1.5, 0*\ystep) {$2:\rd(x,2)$};
  
    \node (t13) at (3.2, 0*\ystep) {$3:\rd(x,1)$};
    
    \node (t21) at (-.2, -2*\ystep) {$1:\wt(x,2)$};
    \node (t22) at (1.5, -2*\ystep) {$2:\wt(y,1)$};

    \node (t31) at (-.2, -4*\ystep) {$1:\rd(y,1)$};

    \draw[po] (t11) -- (t21);
    \draw[po] (t21) -- (t31); 
    
    \draw[po] (t12) -- (t22);
    \end{tikzpicture}}
    
     \caption{The partial interleaving $(1:w(x,1))(3:r(x,1))(1:w(x,2))$ has only one preemption  at $(1:w(x,1))$ while the interleaving $(1:w(x,1))(3:r(x,1))(1:w(x,2)) (2:r(x,2)) (2:w(y,1)) (1:r(y,1))$ has two preemptions at $(1:w(x,1))$ and $(1:w(x,2))$.}\label{fig-onewriter-example}
     
\end{figure}

\noindent{\bf{Classification of programs}} In our analysis, we classify programs based on the number of threads that write to a variable. A concurrent program is said to be \onewriter\ if for every variable $x$, there is a single thread with write operations of the form $\wt(x, d)$. Multiple threads can have read operations $\rd(x, d)$ though. Similarly, a concurrent program is said to be \twowriter\ (resp. \threewriter) if there are at most two (resp. three) threads with write operations to each variable. Figures~\ref{fig-onewriter-example} and \ref{fig:example-algo} give examples of \onewriter\ programs.

\noindent{\bf{Problem Statement.}}
We now reformulate our problem of interest, based on the notions defined above.  

\begin{center}
  \fbox{%
    \begin{minipage}{0.98\textwidth}
      \noindent \textsc{\textbf{VSC$^\pp$-problem}: Verifying Sequential Consistency under $\pp$ preemptions} 

      \medskip

      \noindent \textbf{Fixed:} A constant $\pp \in \mathbb{N}\cup\{0\}$

      \medskip 

      \noindent \textbf{Instance:} A concurrent program $\program$. 

      \medskip

      \noindent \textbf{Question:} Is there an interleaving of $\program$ containing at most $\pp$ preemptions which is sequentially consistent?
    \end{minipage}%
  }
\end{center}
As we show later, the complexity of this problem differs based on whether the program of interest is \onewriter{} or \twowriter{} or \threewriter. 

\section{Polynomial-time algorithm for \onewriter\ programs}
\label{sec:one-writer}
We begin our analysis with \onewriter{} programs, and the  main result of this section 
 is to prove the following theorem. 
\sconewriter*

For the rest of this section, fix a \onewriter\ program $\program$ with $k$ threads having a total of $n$ events across all threads, 
and a constant $\pp \geq 0$. Here is an overview of our algorithm. 
\begin{description}
\item[Step 1.] Guess a set $\points$ of at most $\pi$ events.
\item[Step 2.] Decide whether there exists an SC interleaving of $\program$ where the preemptions occur exactly after the events in $\points$.
\end{description}
Step 1 can be done in $\mathcal{O}(n^{\pp})$.  For Step 2, we present an algorithm that has complexity $\mathcal{O}(n^3 \cdot k^2 \cdot \pp \cdot \pp!)$.  Hence, the overall complexity  comes to $\mathcal{O}(n^{\pp + 3} \cdot k^2 \cdot \pp \cdot \pp!)$. Since $\pi$ is fixed, we get a polynomial-time algorithm for the \vscp\ problem for \onewriter\ programs. We now present an overview of our algorithm through examples. Technical details are presented in Appendix~\ref{app:ptime}.

 \begin{figure}[t]
 \centering
\begin{tikzpicture}[box/.style={rectangle, rounded corners, fill=gray!10, draw}]
  \node [box]  (a) at (0,0) {\scriptsize Program};
  \node [box] (b) at (3,0) {\scriptsize Block-program};
  \node [box] (c) at (6.2,0) {\scriptsize \begin{tabular}{c} Build \textbf{conflict graph} \\ on outer blocks \end{tabular}};
  \node [box] (d) at (9.5,1) {\scriptsize \begin{tabular}{c} Pick permutation \\ of inner blocks \end{tabular}};
  \node [box] (e) at (9.5,-1) {\scriptsize \begin{tabular}{c} Place outer blocks \\ suitably based on \\ conflict order \end{tabular}};

  \begin{scope}[->]
   \draw [>=stealth] (a) to node [above] {\scriptsize \begin{tabular}{c} \textbf{Step 1} \\ choose $\pi$ \\  preemption  \\ points \end{tabular}} (b);
   \draw [thick, dashed, gray] (b) to (c);
   \draw [thick, dashed, gray] (c) to (d);
   \draw [thick, dashed, gray] (d) to (e);
  \end{scope}

  \node at (3, -0.6) {\scriptsize \textbf{Step 2.1}};
  \node at (6.2, -0.8) {\scriptsize \textbf{Step 2.2}};
  \node at (9.5, 1.7) {\scriptsize \textbf{Step 2.3}};
  \node at (9.5, -1.8) {\scriptsize \textbf{Step 2.4}};
\end{tikzpicture}
\caption{Schema for the algorithm for \onewriter\ programs}
\label{fig:schema-one-writer}
 \end{figure}
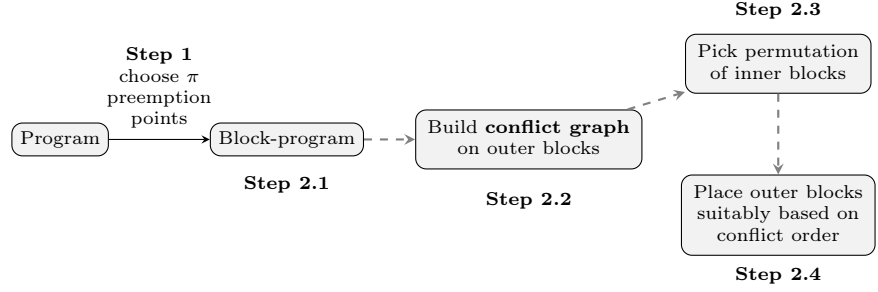

 \subsection{Overview of our procedure}

Figure~\ref{fig:schema-one-writer} presents the overall schema. In Step 1, we choose at most $\pp$ points in the program where preemptions are guessed to occur. Step 2 can be further broken down into 4 steps, which we will now explain using the examples presented in Figure~\ref{fig:example-algo}. In the figure, there are two programs $\program_1$ and $\program_2$. For $\program_1$, Step 1 guesses the preemption points to be $\wt(y,1)$ from $S_2$,  $\wt(x,2)$ from $S_3$ and the first $\rd(x,2)$ from $S_4$. For $\program_2$, the guessed preemption points are $\wt(x,1)$ from $S_1$ and $\wt(y,1)$ from $S_2$.

\noindent{\bf{Step 2.1. A Block-Program from the Blocks.}}  Using these points, the program can be divided into blocks as shown in the right of Figure~\ref{fig:example-algo}. In any interleaving where the preemptions occur at the chosen points, the operations within a block appear contiguously. Therefore, selecting an interleaving with these preemption points amounts to finding a suitable interleaving of the blocks, where sequential consistency is maintained. Observe that when $\pp$ preemption points are chosen there are totally $\pp + k$ blocks (recall that $k$ is the number of threads). We distinguish the blocks into two kinds. The \emph{outer} blocks (coloured in green, in Figure~\ref{fig:example-algo}) are the last blocks in each thread. There are $k$ of them in total. The rest of the blocks are \emph{inner} blocks (coloured in violet, in Figure~\ref{fig:example-algo}). Each inner block ends with a chosen preemption point, and hence there are $\pp$ inner blocks in total.Enumerating all permutations of $\pp + k$ blocks will result in a $(\pp + k)!$ factor in the complexity, which is not polynomial-time when $\pp$ is the only fixed parameter.  To get rid of the factor $k$ inside the factorial, we proceed as follows.

\noindent{\bf{Step 2.2. Order Outer Blocks via Conflict Analysis.}} %
 We show that an ordering among the outer blocks can be determined by a \emph{conflict graph}: outer blocks $b_r$ and $b_w$ are in conflict if there is a variable $x$ for which the last write in $b_w$ is of the form $\wt(x, d')$ and there is a read in $b_r$ of the form $\rd(x, d)$ with $d \neq d'$. If $b_w$ appears before $b_r$,  then $\rd(x,d)$ cannot read the respective value $d$ in any interleaving, since $b_w$ is the last block of the \emph{unique} thread writing on $x$. Hence we draw an edge $b_r \to b_w$ indicating that $b_r$ should be placed before $b_w$. In Figure~\ref{fig:example-algo}, the blue edges denote this ordering induced by the conflict: for $\program_1$, $\rd(x, 2)$ should appear before $\wt(x, 3)$; for $\program_2$, there are two induced orderings due to the different variables $x$ and $y$. Clearly, if the ordering includes a cycle (as in $\program_2$),  there is no sequentially consistent interleaving for this set of blocks. If there is no such cycle (for instance, $\program_1$), we move to Step 2.3.

\begin{figure}[t]
 \centering
 \begin{tikzpicture}
\begin{scope}
 \node at (0,0) {\scriptsize $S_1$};
 \node at (0,-0.6) {\scriptsize $\rd(y,2)$};
 \node at (1.2, 0) {\scriptsize $S_2$};
 \node at (1.2, -0.6) {\scriptsize $\wt(y,1)$};
 \node at (1.2, -1.2) {\scriptsize $\wt(y,2)$};
 \node at (2.4, 0) {\scriptsize $S_3$};
 \node at (2.4, -0.6) {\scriptsize $\wt(x,1)$};
 \node at (2.4, -1.2) {\scriptsize $\wt(x,2)$};
 \node at (2.4, -1.8) {\scriptsize $\wt(x,3)$};
 \node at (3.6, 0) {\scriptsize $S_4$};
 \node at (3.6, -0.6) {\scriptsize $\rd(x, 2)$};
 \node at (3.6, -1.2) {\scriptsize $\rd(x, 2)$};
\end{scope}

\begin{scope}[xshift=6.5cm, outer/.style={rectangle, rounded corners, fill=green!20, draw}, inner/.style={rectangle, rounded corners, fill=violet!20, draw}]
\node at (0,0) {\scriptsize $S_1$};
 \node [outer] at (0,-0.6) {\scriptsize $\rd(y,2)$};
 \node at (1.2, 0) {\scriptsize $S_2$};
 \node [inner] at (1.2, -0.6) {\scriptsize $\wt(y,1)$};
 \node [outer] at (1.2, -1.2) {\scriptsize $\wt(y,2)$};
 \node at (2.4, 0) {\scriptsize $S_3$};
 \draw [rounded corners, fill=violet!20] (1.9, -0.4) rectangle (2.9, -1.4);
 \node  at (2.4, -0.6) {\scriptsize $\wt(x,1)$};
 \node  at (2.4, -1.2) {\scriptsize $\wt(x,2)$};
 \node [outer] (b) at (2.4, -1.8) {\scriptsize $\wt(x,3)$};
 \node at (3.6, 0) {\scriptsize $S_4$};
 \node [inner] at (3.6, -0.6) {\scriptsize $\rd(x, 2)$};
 \node [outer] (a) at (3.6, -1.2) {\scriptsize $\rd(x, 2)$};
\end{scope}

\draw [->, thick, blue] (a) to [bend left] (b);

\draw (-0.5,-2.2) -- (11,-2.2);

\node [left, thick, gray] at (-0.1, 0) {\scriptsize $\program_1$:};
\node [left, thick, gray] at (-0.1, -2.5) {\scriptsize $\program_2$:};

\begin{scope}[yshift=-2.5cm]
\node at (1, 0) {\scriptsize $S_1$};
\node at (1, -0.6) {\scriptsize $\wt(x,1)$};
\node at (1, -1.2) {\scriptsize $\rd(y,1)$};
\node at (1, -1.8) {\scriptsize $\wt(x,2)$};

\node at (2.6, 0) {\scriptsize $S_2$};
\node at (2.6, -0.6) {\scriptsize $\wt(y,1)$};
\node at (2.6, -1.2) {\scriptsize $\rd(x,1)$};
\node at (2.6, -1.8) {\scriptsize $\wt(y,2)$};

\end{scope}

\begin{scope}[xshift = 6.5cm, yshift=-2.5cm, outer/.style={rectangle, rounded corners, fill=green!20, draw}, inner/.style={rectangle, rounded corners, fill=violet!20, draw}]
\node at (1, 0) {\scriptsize $S_1$};
\node [inner] at (1, -0.6) {\scriptsize $\wt(x,1)$};
\draw [rounded corners, fill = green!20] (0.5, -1) rectangle (1.5, -2);
\node (c) at (1, -1.2) {\scriptsize $\rd(y,1)$};
\node (e) at (1, -1.8) {\scriptsize $\wt(x,2)$};

\node at (2.6, 0) {\scriptsize $S_2$};
\node [inner] at (2.6, -0.6) {\scriptsize $\wt(y,1)$};
\draw [rounded corners, fill = green!20] (2.1, -1) rectangle (3.1, -2);
\node (d) at (2.6, -1.2) {\scriptsize $\rd(x,1)$};
\node (f) at (2.6, -1.8) {\scriptsize $\wt(y,2)$};

\draw [thick, blue, ->] (c) to [bend left] (d);
\draw [thick, blue, ->] (f) to [bend left] (e);
\end{scope}

\end{tikzpicture}
\caption{Programs $\program_1$ and $\program_2$ on the left, their corresponding block-programs on the right w.r.t. chosen preemption points. Inner blocks are coloured in violet, and outer blocks in green. The blue edges denote edges in the conflict graph of the outer nodes.}
\label{fig:example-algo}
\end{figure}
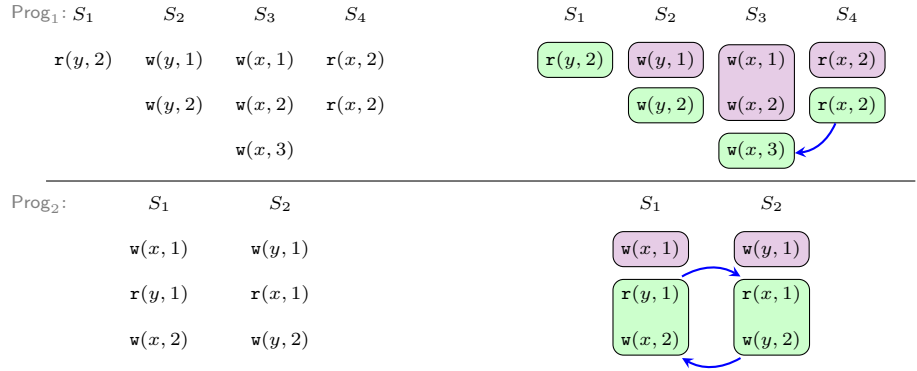

\begin{figure}[t]
\centering
\begin{tikzpicture}[outer/.style={rectangle, rounded corners, fill=green!20, draw}, inner/.style={rectangle, rounded corners, fill=violet!20, draw}]

 \node [inner] (a1) at (0,1.5) {\scriptsize $\wt(x,1)~~\wt(x,2)$};

  \node at (-1.2, 1.6) {\scriptsize $0$};

  \node [inner] (b1) at (2,1.5) {\scriptsize $\wt(y,1)$};

  \node at (1.3, 1.6) {\scriptsize $1$};

  \node at (2.8, 1.6) {\scriptsize $2$};

  \node [inner] (c1) at (6.5,1.5) {\scriptsize $\rd(x, 2)$};
 
  \node at (7.2, 1.6) {\scriptsize $3$};

  \node [inner] (a) at (0,0) {\scriptsize $\wt(x,1)~~\wt(x,2)$};
  \draw [thick] (-1.2, 0) -- (-1.2, -0.2);
  \node at (-1.2, -0.4) {\scriptsize $0$};

  \node [inner] (b) at (2,0) {\scriptsize $\wt(y,1)$};
\draw [thick] (1.3, 0) -- (1.3, -0.2);
  \node at (1.3, -0.4) {\scriptsize $1$};

  \draw [thick] (2.8, 0) -- (2.8, -0.2);
  \node at (2.8, -0.4) {\scriptsize $2$};

  \node [inner] (c) at (6.5,0) {\scriptsize $\rd(x, 2)$};
  \draw [thick] (7.2, 0) -- (7.2, -0.2);
  \node at (7.2, -0.4) {\scriptsize $3$};

  \node [outer] at (3.5, 0) {\scriptsize $\wt(y, 2)$};
  \node [outer] at (5, 0) {\scriptsize $\rd(y, 2)$};
  \node [outer] at (8, 0) {\scriptsize $\rd(x,2)$};
  \node [outer] at (9.5, 0) {\scriptsize $\wt(x,3)$};

\end{tikzpicture}
\caption{A permutation of inner blocks (above), and an interleaving obtained by 
inserting outer blocks at appropriate positions (below).}
\label{fig:example-algo-final-step}
\end{figure}
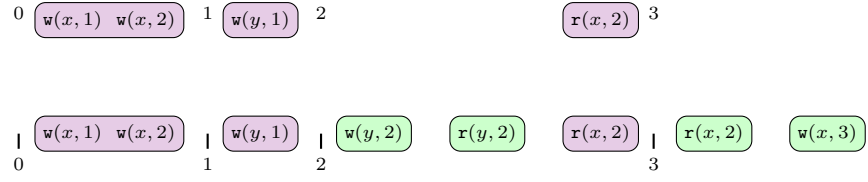

\noindent{\bf{Step 2.3. Insert Outer Blocks into a Permutation of Inner Blocks.}} Consider all possible permutations of the inner blocks. Since there are at most $\pp$ in number, there are at most $\pp!$ permutations. The goal is to check if we can place the $k$ outer blocks between these inner blocks so as to obtain an SC interleaving. In Figure~\ref{fig:example-algo-final-step}, the top illustrates an ordering for inner blocks.  We have also marked the positions around the inner blocks (in the figure, they are marked as $0, 1, 2$ and $3$). The outer blocks need to be placed appropriately.

\noindent{\bf{Step 2.4. Placement of Outer Blocks respecting conflicts.}} Continuing with our running example $\program_1$ in Figure~\ref{fig:example-algo-final-step}, at position $0$, none of the outer blocks is enabled: the outer block in $S_1$ requires a writing source, and the outer blocks in the other threads cannot be scheduled unless earlier blocks in the program order are completed. So we move to position $1$. At position $1$, the outer block $\wt(x, 3)$ of $S_3$ is enabled since the inner block $\wt(x,1).\wt(x,2)$ of $S_3$ 
appears already.  However, we cannot place it at position 1 for two reasons: firstly, it conflicts the inner block $\rd(x, 2)$ that appears in front of position 3 to the right; and secondly there is an outer $\rd(x,2)$ block which needs to appear before $\wt(x, 3)$ according to the ordering induced by the conflict. Therefore, we do not place it at position $1$ and move on to position $2$. At position $2$, the enabled outer blocks are $\wt(y, 2)$ and $\wt(x, 3)$. We place $\wt(y, 2)$ and not $\wt(x, 3)$ for the same reasons as above. Once $\wt(y, 2)$ is placed, the outer block $\rd(y, 2)$ is also enabled. We place it and move on to position $3$. Here, both $\wt(x, 3)$ and the outer block $\rd(x, 2)$ are enabled. Based on the conflict ordering, we place $\rd(x, 2)$ followed by $\wt(x, 3)$. 

Thus, at each position in a given permutation of inner blocks,
 we find the enabled outer blocks among the ones yet to be scheduled, and  place them according to the ordering induced by the conflicts. Moreover, we place an outer block at a position only if it does not conflict any inner block appearing later. Placing an outer block at a position may enable more outer blocks. We finish with the position only when there are no more outer blocks that can be placed at that position.  

\noindent{\bf{A final check}.} If some outer blocks cannot be placed, we conclude that an SC interleaving is not possible for the chosen permutation of the inner blocks. Otherwise, all outer blocks have been placed. Still, we need to do a final check to ensure SC. This is because, we had chosen an arbitrary permutation of inner blocks. There could be reads in inner blocks with no writing source, even after the placement of outer blocks; and secondly there could be reads in inner blocks which contradict later reads in an inner block, even after the outer blocks are all placed. If the interleaving satisfies SC, we are done. Else, we move to another permutation of inner blocks, until we exhaust all of them.

\noindent{\bf{Complexity}.} Recall that $n$ is the number of events, $k$ is the number of threads and $\pp$ the number of preemptions. Step 2.1 takes $\mathcal{O}(n)$. Step 2.2 builds a conflict ordering among $k$ outer blocks and checks for a cycle. This process takes at most $\mathcal{O}(k \cdot n)$. Once an enumeration of inner blocks is chosen, Step 2.4 can be done in $\mathcal{O}(\pp \cdot k \cdot n)$: at each position, we manipulate at most $k$ blocks, and for each block we may do an $\mathcal{O}(n)$ analysis. The final check takes $\mathcal{O}(n)$, and hence Step 2.4 requires an overall complexity $\mathcal{O}(\pp \cdot k \cdot n)$. Since there are $\pp!$ enumerations, the complexity of Step 2 comes to $\mathcal{O}(n +  k \cdot n  +  \pp! \cdot \pp \cdot k \cdot n)$, which is $\mathcal{O}(n \cdot k \cdot \pp \cdot \pp!)$. Combining with the $\mathcal{O}(n^\pi)$ complexity of Step 1, we get an overall complexity of $\mathcal{O}(n^{\pp+1} \cdot k\cdot \pp \cdot \pp!)$.

\subsubsection*{Why the conflict analysis does not work for inner blocks and \twowriter\ programs?}

Let us once again consider the $\program_1$ example from Figure~\ref{fig:example-algo}. There $\rd(x,2)$ is in conflict with $\wt(x,3)$. If $\wt(x, 3)$ is not an outer block, it may still be possible to have an interleaving where $\wt(x,3)$ appears before $\rd(x,2)$ as long as there is a later block below $\wt(x,3)$ with a matching $\wt(x,2)$. So, there is no forced order. A similar phenomenon happens for \twowriter\ systems too. We cannot force an order between two blocks just because there is conflict in the read and write. This leads to a blow-up in the permutations, and as we show in the next section, when there are two-writers, the \vscp\ problem is $\NP$-hard already with $\pp = 0$, exemplifying the crucial advantage of the conflict analysis in \onewriter\ systems.

We present the formal details of our algorithm in Appendix~\ref{app:ptime}.

\section{Hardness results}
\label{sec:hardness}
 In Section~\ref{sec:one-writer}, we showed that the \vscp\ problem for \onewriter\ programs can be solved in polynomial-time, with a complexity $\mathcal{O}(n^{\pp + 1} \cdot k \cdot \pp \cdot \pp!)$. Here are two observations about this complexity: the exponent on $k$ is a constant, whereas the exponent on $n$ is the parameter $\pp$. We remark that to achieve the polynomial-time complexity, we had restricted to \onewriter\ programs and also assumed that the parameter $\pp$ is fixed. Our goal in this section is to investigate what happens when we perturb these restrictions. Figure~\ref{fig:results-overview} presents an overview of the complexity results.

\begin{figure}[h]
\centering
\begin{tikzpicture}[box/.style={rectangle, fill=gray!10, draw, inner sep=2pt}]
  \node[box] (a) at (0,0) {\footnotesize \begin{tabular}{c} \onewriter \\ fixed $\pp$ \end{tabular}};
  \node[box] (b) at (-4,-1.5) {\footnotesize \begin{tabular}{c} multiple writers \\ fixed $\pp$ \end{tabular}};
  \node[box] (c) at (0, -1.5) {\footnotesize \begin{tabular}{c} \onewriter \\ $\pp$ not fixed \end{tabular}};
  \node[box] (d) at (4, -1.5) {\footnotesize \begin{tabular}{c} multiple writers \\ $\pp$ not fixed \end{tabular}};
  
  \begin{scope}[->]
    \draw (a) to (b);
    \draw (a) to (c);
    \draw (a) to (d);
  \end{scope}

  \node at (2.5, 0) {\footnotesize \textsc{Ptime} (Section~\ref{sec:one-writer})};
  \node at (-4, -2.5) {\footnotesize \begin{tabular}{c} \textsc{NP-complete}  \\ (Sections~\ref{sec:finegrained},\ref{sec:np}) \end{tabular}};
  \node at (0, -2.5) {\footnotesize \textsc{NP-complete}~\cite{Gibbons}};
  \node at (4, -2.5) {\footnotesize \begin{tabular}{c} \textsc{W[1]-hard} \\ (Section~\ref{sec:w1-hard}) \end{tabular}};
\end{tikzpicture}
\caption{Illustration of the various complexity results}
\label{fig:results-overview}
\end{figure}
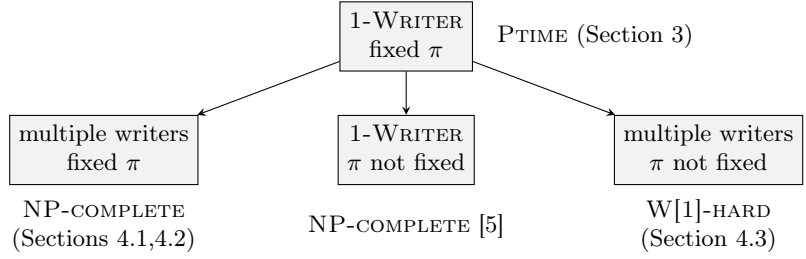

\paragraph*{Multiple writers, fixed $\pp$.} When multiple writers are allowed, \vscp\ becomes $\NP$-complete, even with \twowriter\ programs (Theorem~\ref{th:NPhardness}). Furthermore, under the Exponential-Time-Hypothesis~\cite{ImpagliazzoPZ01}, we prove that \vscp\ cannot have a $2^{o(k)} \cdot n^{\mathcal{O}(1)}$ algorithm, even for $0$ preemptions and \threewriter\ programs (Theorem~\ref{th:Finegrained}). Contrast this with \onewriter\ where the exponent on $k$ was a constant $1$, thanks to the conflict analysis. Our hardness results therefore give strong evidence that the presence of multiple writers may introduce an unavoidable exponential on the number of threads $k$. 

\paragraph*{One writer, $\pp$ not fixed.} Now, if we continue with \onewriter\ but assume that $\pp$ is also part of the instance input, then Theorem~4.3 of \cite{Gibbons} shows $\NP$-completeness. In their proof, an instance of 3-CNF-SAT with $k$ variables and $m$ clauses is reduced to a \onewriter\ program, where $\pp$ can be taken to be $\mathcal{O}(m + k)$. This is sufficient to prove $\NP$-hardness, when $\pp$ is not fixed. This hardness result provides evidence that having either $n^\pp$ or $k^\pp$ or some exponential function with $\pp$ could be unavoidable -- in our algorithm for Section~\ref{sec:one-writer}, we have $n^\pp$ and $\pp!$ in the running time. 

\paragraph*{Multiple writers, $\pp$ not fixed.} Finally, when we consider multiple writers, and do not fix $\pp$, we show W[1]-hardness (Section~\ref{sec:w1-hard}). W[1] is a complexity class studied in parameterized complexity, which intuitively provides evidence that we cannot have an algorithm with running time $f(\pp) \cdot (k + n)^{\mathcal{O}(1)}$ for our problem when multiple writers are allowed. In other words, we cannot separate out the influence of $\pp$ from $n$ and $k$. In contrast, consider the standard \emph{vertex cover} problem for graphs: given a graph of size $n$ is there a vertex cover of size $\le k$? This is an $\NP$-complete problem, but it has an algorithm of the form $f(k) \cdot n^{\mathcal{O}(1)}$ and hence is called fixed-parameter-tractable (FPT). On the other hand, deciding whether a graph of size $n$ has an \emph{independent set} of size $\ge k$ is known to be W[1]-hard (hence unlikely to be FPT). In Section~\ref{sec:w1-hard}, we reduce the independent set problem to our VSC problem, where the $k$ of the independent set translates to the $\pp$ of our VSC problem. We refer the reader to~\cite{DowneyFellows1999} for an exposition of parameterized complexity classes.

We now move on to present the hardness results. Firstly, we note that \vscp\ is in $\NP$, since we can guess an interleaving and check sequential consistency and the preemption bound $\pp$ in linear time. 
We present two hardness results for the case of fixed $\pp$. To show hardness, we provide reductions from 3-CNF-SAT. Given a Boolean formula $\varphi$ over variables $x_1, \dots, x_k$ and clauses $C_1, \dots, C_m$, where each clause is a disjunction of three literals, we ask whether $\varphi$ admits a satisfying assignment. We will start with the reduction to \threewriter\ systems and then adapt the reduction to \twowriter\ systems. The constructions for both the reductions are adapted from Theorem 2.1 of~\cite{Gibbons}. Let us first recall the Exponential Time Hypothesis that we make use of, to deduce a conditional lower bound.

\textbf{Exponential Time Hypothesis (ETH).}~\cite{ImpagliazzoPZ01}
There is no $2^{o(k)}\cdot (k+m)^{\mathcal{O}(1)}$ algorithm for 3-CNF-SAT, where $k$ is the number of variables and $m$ is the number of clauses.

For the reduction to $\threewriter$, we construct a concurrent program consisting of $\mathcal{O}(k)$ threads, i.e., linear in the number of variables of the 3-CNF-SAT instance, and no dependence on the number of clauses $m$. This yields a parameterized reduction and implies the desired conditional lower bound under ETH. In contrast, the reduction to $\twowriter$ uses $\mathcal{O}(k+m)$ threads and is therefore, not a parameterized reduction. Sections~\ref{sec:finegrained} and~\ref{sec:np} describe these reductions in detail.

\paragraph{Notation.}
Let $\varphi$ be a 3-CNF formula over variables $x_1, \dots, x_k$ and clauses $C_1, \dots, C_m$. Each clause $C_j$ consists of three literals $(\ell_j^1, \ell_j^2, \ell_j^3)$, where each literal is either a variable $x_i$ or its negation $\neg x_i$.

For each variable $x_i$, define:
\[
S_{x_i} := \{ j \mid C_j \text{ contains } x_i \}, \qquad
S_{\neg x_i} := \{ j \mid C_j \text{ contains } \neg x_i \}.
\]

Note that for ease of presentation, we omit thread identifiers in events when they are clear from the context.

\subsection{Fine-grained Hardness under ETH for \threewriter}\label{sec:finegrained}

Given a 3-CNF-SAT instance $\varphi$ in $k$ variable and $m$ clauses, we  construct a \vscp~ instance $\program$ in $\mathcal{O}(k)$ threads and $\mathcal{O}(k+m)$ operations, in $\mathcal{O}(k + m)$ time, such that for every variable only three threads write to it. Therefore, a $2^{o(k)} \cdot n^{\mathcal{O}(1)}$ algorithm for \vscp\ would contradict ETH.
 
\scFinegrained*

Say $\varphi=C_1\land C_2\land \dots \land C_m$. 
The construction for the reduction from $\varphi$, illustrated in Table~\ref{tab:reduction-threads-scmm} shows a program $\program_{\varphi}$ with $6k+1$ threads and $k+m+1$ variables. Variables $v_i, i=1,2,\dots,k$ correspond to $k$ variables of $\varphi$ while variables $c_i, i=1,2,\dots m$ correspond to the clauses. The variable $x$ is a guard to ensure all clause variables $c_i$ are set. Variable $x$ is written only once in $S_f$ while threads writing to $v_i$ are $S_i^0$ and $S_i^1$. Variable $c_i$ is written by $S_{x_i}$ if the clause $C_i$ contains $x$. If clause $C_i$ contains $\neg v_i$, then $c_i$ is written by $S_{\neg x_i}$. As a clause can have at most 3 literals in a 3-CNF-SAT instance, each $c_i$ will be written by at most 3 threads (hence we get a \threewriter\ program). We also pick $\pp$ to be $0$.

\begin{table}[h]
    \caption{Threads in the reduction to \threewriter\ programs}
    \label{tab:reduction-threads-scmm}
    \noindent
    \begin{minipage}{0.38\textwidth}
    \centering
    \begin{tabular}{c | c | c | c}
        $S_i^0$ & $H_i^0$ & $S_i^1$ & $H_i^1$ \\
        $\wt(v_i, 0)$ & $\rd(x, 1)$ & $\wt(v_i, 1)$ & $\rd(x, 1)$ \\
        & $\rd(v_i, 0)$ & & $\rd(v_i, 1)$ 
    \end{tabular}
    \end{minipage}
    \hfill
    \begin{minipage}{0.32\textwidth}
    \centering
    \begin{tabular}{c | c }
        $S_{x_i}$ & $S_{\neg x_i}$ \\
        $\rd(v_i, 1)$ & $\rd(v_i, 0)$ \\
        $\wt(c_j, 1)$ & $\wt(c_j, 1)$ \\
        $\forall j \in C_{x_i}$ & $\forall j \in C_{\neg x_i}$ 
    \end{tabular}
    \end{minipage}
    \hfill
    \begin{minipage}{0.22\textwidth}
    \centering
    \begin{tabular}{c}
        $S_f$ \\
        $\rd(c_1, 1)$ \\
        $\cdots$ \\
        $\rd(c_m, 1)$ \\
        $\wt(x, 1)$
    \end{tabular}
    \end{minipage}
\end{table}

Here is the main idea. In any (complete) $0$-preemption interleaving of $\program_\varphi$, the \emph{helper} threads $H_i^0$ and $H_i^1$ can be executed only after $S_f$ (due to the read-write on variable $x$). This ensures that only one of $S_i^0$ or $S_i^1$ can be executed before $S_f$: otherwise, if both are executed above $S_f$, only the helper thread corresponding to the last write on $v_i$ can be executed after $S_f$. Therefore, since a unique thread between $S_i^0$ or $S_i^1$ appears before $S_f$, we obtain an assignment for the variables. Moreover, as all the clauses have been set to true (thanks to $\rd(c_i, 1)$ in $S_f$), we can conclude that the assignment is satisfying. This leads to the following lemma, whose proof appears in Appendix~\ref{app:hardness}.

\begin{restatable}{lemma}{threeWriterLemma}\label{lem:eth-sc}
    $\varphi$ is satisfiable iff there exists a sequentially consistent partial interleaving $\sigma$ of $\program_\varphi$ with $0$ preemptions.
\end{restatable}

\subsection{NP-Hardness for \twowriter}\label{sec:np}

Here, our objective is to prove hardness even when at most two threads write to each variable. We reduce 3-CNF-SAT to \vscp\ for \twowriter\ programs. In this reduction, a formula $\varphi$ with $k$ variables and $m$ clauses yields a program $\program_\varphi$ with $\mathcal{O}(k+m)$ threads. While this suffices to show $\NP$-hardness, it does not yield the stronger ETH-based lower bound, which requires $\mathcal{O}(k)$ threads (and that does not depend on the number of clauses $m$).

The key challenge here is to reduce the number of threads writing to each clause variable $c_j$ from three (one for each literal in the clause) to at most two, while still encoding the satisfiability of the 3-CNF formula.
For each clause $C_j = (\ell_j^1 \vee \ell_j^2 \vee \ell_j^3)$, we introduce an additional auxiliary variable $d_j$, and modify the construction as follows.

\begin{itemize}[left=0.5em]
\item In the thread corresponding to the third literal $\ell_j^3$, replace the write $\wt(c_j,1)$ with $\wt(d_j,1)$.
\item Introduce a new thread with two events as follows:
\[
S_j = (j : \rd(c_j,1)) \cdot (j : \wt(d_j,1)).
\]
\item In the thread $S_f$, replace each $\rd(c_j,1)$ with $\rd(d_j,1)$.
\end{itemize}

This results in the same outcome. Either clause $C_j$ is set to true due to the third literal (and hence $\wt(d_j, 1)$ happens), or if it needs to be true due to the first or second literal, then $S_j$ comes into effect and performs $\wt(d_j, 1)$.
Let $\program_\varphi$ denote the resulting program. Note that due to the additional threads $S_j$, the total number of threads now also depends on $m$, the number of clauses.

\begin{restatable}{lemma}{twoWriterLemma}\label{lem:np-two-writer}
$\varphi$ is satisfiable iff there exists a sequentially consistent partial interleaving $\sigma$ of $\program_\varphi$ with $0$ preemptions.
\end{restatable}

Thus, we have shown that the preemption-bounded consistency problem is $\NP$-hard for \twowriter~programs.

\subsection{W[1]-hardness for \vscp\ when $\pp$ is a parameter and not fixed}
\label{sec:w1-hard}

In this section, we will show that the  VSC problem is $W[1]$-hard when the number of preemptions $\pp$ is an input parameter.

\wonehardness*

The proof follows from a parameterized reduction from $k$-independent set problem. Let $G = (V,E)$ be a connected graph with vertex set $V$ and edge set $E$. The $k$-independent set problem asks if there exists a set of vertices $V_k \subseteq V$ of size $k$ such that no two vertices share a common edge. 

We will reduce an instance $(G, k)$ of the independent set problem to an instance $(\program_G, \pp)$ with $\pp = 3k$, of the \vscp\ problem. The program $\program_G$ is depicted in Table~\ref{tab:independent-set}. It has an $\Init$ thread, $k$ checker threads $\Checker_1, \Checker_2, \dots, \Checker_k$, and $k$ selector threads $\Sel_1, \Sel_2, \dots, \Sel_k$. There is a variable $y_e$ corresponding to each edge $e \in E$, variables $x_1, x_2, \dots, x_k$ to help pick $k$ vertices, and some book-keeping variables $s, p_0, \dots, p_k$. 

\begin{table}
\centering
\caption{Reduction from $k$-independent set problem. Index $j$ ranges from $1$ to $k$.}
\label{tab:independent-set}
\medskip
\centering
\begin{tabular}{l  @{\qquad} l  @{\qquad} l }
    $\Init$ &  $\Checker_j$ & $~\Sel_j$ \\
    &  & \\
$\wt(y_e, 0)$   &  $\rd(p_{j-1}, 1)$ &   $\forall v \in V:$  \\
$\forall e \in E$  &  $\wt(s, 0)$ when $j \neq k$ &  \\
& & $\rd(y_e, 0);~\wt(y_e, j)$ \quad $\forall e \in E$ containing $v$ \\
$\wt(x_j, 1)$ &  $\rd(x_j, 0)$  &  \\
$\forall 1 \le j \le k$ & $\wt(s, 1)$ when $j = k$  & $\rd(s, 1)$   \\
  & &  $\wt(x_j, 0)$   \\
$\wt(s, 1)$ & $\wt(p_j, 1)$ &  $\wt(x_j, 1)$  \\
$\wt(p_0, 1)$ & &   \\
 & &   $\rd(y_e, j);~\wt(y_e, 0)$ \quad $\forall e \in E$ containing $v$  
\end{tabular}
\end{table}

The $\Init$ thread is self-explanatory from Table~\ref{tab:independent-set}. For the first $k-1$ checker threads $\Checker_j$ (with $1 \le j \le k-1$) the second operation is $\wt(s, 0)$ (depicted as $\wt(s,0)$ when $j \neq k$). These threads do not have a $\wt(s,1)$ instruction at all. For the last checker thread $\Checker_k$, there is no $\wt(s,0)$, and instead there is a $\wt(s, 1)$ right after $\rd(x_j, 0)$. 

We now explain the selector threads. Each selector thread $\Sel_j$  
has $|V|$ blocks of operations, one block for each $v \in V$. Taking $V = \{u_1, \dots, u_n\}$, the selector thread $\Sel_j$ can be seen as:
\begin{align*}
 \Sel_j: \quad \BB^j_{u_1}~ \BB^j_{u_2}~\cdots~ \BB^j_{u_n}
\end{align*} 
where each $B^j_u$ is a block of operations corresponding to a vertex $u$. This can further be seen as sub-blocks:
\begin{align*}
 \BB^j_u:~~&~  \Start^j_u ~~\rd(s, 1) ~~ \wt(x_j, 0) ~~ \wt(x_j, 1) ~~ \Finish^j_u \\
 \text{ where~} \qquad \qquad \Start^j_u: ~~& ~\rd(y_e, 0) ~ \wt(y_e, j) ~~ \forall e  \text{ incident on } u \\
\Finish^j_u: ~~& ~\rd(y_e, j) ~ \wt(y_e, 0) ~~ \forall e \text{ incident on } u
\end{align*}

\paragraph*{Overview of the W[1]-hardness reduction.}

Consider an SC interleaving of $\program_G$. It begins with the execution of the $\Init$ thread. The checker and selector threads appear later. 
Notice two conflicting events: the $\Init$ thread has $\wt(x_j, 1)$ whereas $\Checker_j$ has a read $\rd(x_j, 0)$. Therefore, a $\wt(x_j, 0)$ has to appear somewhere in between. This can be provided by $\Sel_j$: for $\wt(x_j, 0)$ to be the last write between the conflicting events in $\Init$ and $\Checker_j$, a part of $\Sel_j$ ending at a $\wt(x_j,0)$ operation should appear in between. So, we say that $\Sel_j$ is ``inside'' some block corresponding to a vertex -- call this vertex $v_j$. Intuitively, we require every $\Sel_j$ to have executed upto a $\wt(x_j, 0)$, then wait for $\Checker_j$ to finish, and subequently execute the rest of $\Sel_j$ that starts with $\wt(x_j, 1)$. However, we cannot do this na\"ively. Suppose $\Sel_1$ executes upto $\wt(x_1,0)$ of $\BB^1_{v_1}$. Then at the start of this block, we have writes $\wt(y_e, 1)$ for all $e$ incident on $v_1$. This may disable the execution of $\Sel_2$ -- notice that the beginning of each block contains a read of the form $\rd(y_e, 0)$. Therefore we need to carefully reach to $\wt(x_j, 0)$ in each thread. We do this in phases. 

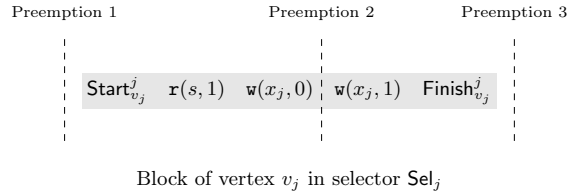
\begin{figure}[h]
\centering
\scalebox{0.85}{\begin{tikzpicture}[box/.style={rectangle, fill=gray!20, inner sep=2pt}]
\begin{scope}[every node/.style={box}]
  
  \node (b) at (4,0) {\footnotesize $\Start^j_{v_j}~$ $~\rd(s, 1)~$ $~\wt(x_j,0)~$ $~\wt(x_j, 1)~$ $~\Finish^j_{v_j}$ };
\end{scope}
\node at (4, -1.4) {\footnotesize Block of vertex $v_j$ in selector $\Sel_j$};
\draw [dashed] (0.5, 0.8) -- (0.5, -0.8);
\draw [dashed] (4.5, 0.8) -- (4.5, -0.8);
\draw [dashed] (7.5, 0.8) -- (7.5, -0.8);

\node at (0.5, 1.2) {\scriptsize  Preemption $1$};
\node at (4.5, 1.2) {\scriptsize  Preemption $2$};
\node at (7.5, 1.2) {\scriptsize  Preemption $3$};
\end{tikzpicture}}
\caption{Illustrating the preemption points within a selector thread}
\label{fig:w1-illustration}
\end{figure}

First we make all the selector threads execute upto the beginning of some block (see preemption $1$ in Figure~\ref{fig:w1-illustration}), one after the other. Since the end of each block ``resets'' $y_e$ to $0$ (through the $\wt(y_e, 0)$ instruction), such an execution is possible. As a second step, we make every selector execute upto $\wt(x_j, 0)$, that is, between preemption points $1$ and $2$ in Figure~\ref{fig:w1-illustration}. The crucial point is that this phase of executions is possible only when the selectors have picked an independent set: on executing $\Sel_j$ upto preemption point $2$, all $y_e$ for edges $e$ incident on $v_j$ have been set to $j$, but since none of these edges is incident on a $v_i$ ($i \neq j$), the selector thread $\Sel_i$ can also pass through its $\rd(y_{e'}, 0)$ operations and execute upto $\wt(x_i, 0)$. 

After all selectors are executed upto preemption point $2$, the checker threads can be executed. Finally, we want to execute the rest of the selector threads after the checker threads are completed. However, executing the full selector may not be possible since some of the $y_e$ variables have current values different from $0$. So to once again reset all $y_e$ to $0$, we execute each selector upto preemption point $3$ shown in Figure~\ref{fig:w1-illustration}. Then the rest of the operations are executed in each selector thread, one by one, until each one completes. 
This shows that when the selectors choose an independent set, an SC interleaving with at most $3k$ preemptions is possible, with $3$ preemptions occuring in each selector thread. 

For the converse, we can show that if there is an SC interleaving, then the preemptions within the selector threads encode an independent set. We start with the fact that in the SC interleaving, a $\wt(x_j, 0)$ should appear before the $\rd(x_j, 0)$ of $\Checker_j$. Hence $\Sel_j$ should have a preemption inside a vertex-block at $\wt(x_j, 0)$, after which $\Checker_j$ happens, and then the rest of $\Sel_j$.   Here is an important detail to take care of: in order to enforce that the selected threads form an independent set, we need a point during the execution where all the selector threads have simultaneously ``inside'' a block. One possible way this may not happen is when $\Sel_j$ executes the operations between preemption points $2$ and $3$ right after $\Checker_j$ finishes, and then $\Sel_{j+1}$ enters inside a block. Then $\Sel_j$ and $\Sel_{j+1}$ are not simultaneously ``inside''.  To avoid this, we add a $\wt(s, 0)$ in $\Checker_j$ and a $\rd(s, 1)$ in $\Sel_{j+1}$ that will prevent $\Sel_{j+1}$ entering inside, after $\Checker_j$. Formal proofs for this reduction can be found in Appendix~\ref{app:w1-hard}.

\section{Conclusion}

In this work, we have presented a study of the fundamental VSC-problem in the context of bounded preemptions. This is a deviation from the usual parameters considered in the study of this problem. Our results are two-fold: an algorithm for the single-writer case, and hardness reductions in the presence of multiple writers, and unbounded preemptions. 
On the algorithmic front, the key insight is that for \onewriter\ systems the idea of a conflict graph helps to determine ordering without having to fall back on a systematic enumeration of all permutations. Such a trick is not possible in the case of multiple writers. 

Investigating a preemption-bounded version of stateless model-checking under the reads-value-from equivalence and the impact of incorporating the \vscp\ problem into the procedure, remain directions for future work. Another interesting direction is to investigate the consistency problem for the Total Store Ordering (TSO) memory model under the influence of various context-switches~\cite{TSO-context-bounding} defined in the literature. 

\bibliographystyle{splncs04}
\bibliography{m}

\appendix

\newpage 

\section{Appendix for Section~\ref{sec:one-writer}}
\label{app:ptime}

\subsection{Step 2.1. Partition the program $\program$ into Blocks : Block Programs}

A subset of events $\points \subseteq \eventset$ represents \emph{preemption points} if no event in $\points$ is the last event of its thread. Assume we have chosen a specific set $\points$ of preemption points. Note that when $\pi=0$, $\points$ is the emptyset.

Let $c^\points_i \ge 0$ be the number of events of thread $i$ in $\points$. 
We will omit the superscript and simply write $c_i$ when $\points$ is clear from the context. Furthermore,  $\sum_{i} c_i \le \pp$. Each thread $S_i$ can then be divided into blocks, based on the events of $S_i$ that are in $\points$, as explained next.

Consider thread $S_i$ as 

 $S_i = a_1 a_2 \dots a_m$.

Assume we guess events $a_{j_1}, a_{j_2}, a_{j_3}$ with $1 \le j_1 < j_2 < j_3 < m$ to be the events after which preemptions occur (in general, we guess $c_i$ events in thread $S_i$; we have taken $c_i = 3$ just as an example). This divides the thread into four blocks:
\begin{align*}
b^i_1: a_1 \dots a_{j_1} \qquad b^i_2: a_{j_1 + 1} \cdots a_{j_2} \qquad b^i_3: a_{j_2 + 1} \cdots a_{j_3} \qquad b^i_4: a_{j_3 + 1} \cdots a_m
\end{align*}
Each block is a contiguous set of events of a thread, and in every interleaving where the preemption points for $S$ occur at $a_{j_1}, a_{j_2}, a_{j_3}$, the events within a block appear contiguously, and moreover no two consecutive blocks of a thread appear next to each other. 

We denote the thread $S_i$ partitioned into blocks as $S_i^\points$. 
\begin{align*}
    \thread^\points_i = b^i_1 \cdot b^i_2 \cdots b^i_{c_i + 1}
\end{align*}
where each block $b^i_j$ is a contiguous sequence of events of $\thread_i$ ending in a preeemption point, as explained above. We call $\thread^\points_i$ as a \emph{block-thread}, and the set of block-threads $\thread^\points_1, \thread^\points_2, \dots, \thread^\points_k$ a \emph{block-program}, which we denote by $\program^\points$.

\begin{example}\label{eg:block-program}
In Figure~\ref{fig-onewriter-example} if $\points_1=\{(1:w(x,1)), (2:r(x,2))\}$ be the preemption points, then the block-program w.r.t. $\points_1$ would look as follows: 
\begin{alignat*}{2}
    \thread^{\points_1}_1 &=b^1_1\cdot b_2^1 &&=(1:w(x,1)) ~\cdot~ (1:w(x,2))(1:r(y,1))\\
     \thread^{\points_1}_2 
     &= b^2_1\cdot b_2^2 &&=(2:r(x,2)) ~\cdot~ (2:w(y,1))\\
      \thread^{\points_1}_3 &=b^3_1 &&=(3:r(x,1))
\end{alignat*}
where the first thread $\thread^{\points_1}_1$ is partitioned as  $b_1^1=(1:w(x,1))$, $b_2^1=(1:w(x,2))(1:r(y,1))$, and similarly for the other threads. 

Likewise, for the set $\points_2=\{(1:w(x,1)), (1:w(x,2))\}$ of preemption points, the block-program w.r.t. $\points_2$ would look as follows: 
\begin{alignat*}{2}\label{blockpartition}
    \thread^{\points_2}_1 &= b^1_1\cdot b_2^1\cdot b_3^1 &&=(1:w(x,1)) ~\cdot~(1:w(x,2))~\cdot~ (1:r(y,1))\\
     \thread^{\points_2}_2 
     &=b_1^2 &&=(2:r(x,2))(2:(w(y,1)))\\
      \thread^{\points_2}_3 &=b_1^3 &&= (3:r(x,1))
\end{alignat*}
Finally, for the set $\points_3=\emptyset$, $\thread^{\points_3}_1{=} b^1_1 {=} \thread_1$,  
$\thread^{\points_3}_2{=} b^2_1 {=} \thread_2$ and $\thread^{\points_3}_3{=} b^3_1 {=} \thread_3$ consist of a single block each.
\end{example}

\subsection*{Sequential Consistency for Block Programs} We extrapolate the notion of sequential consistency in a natural way to the obtained block-program. A partial interleaving $\sigma^\points$ of $\program^\points$ is sequentially consistent if the partial interleaving $\sigma$ of the original program $\program$ obtained by expanding each block of $\sigma^\points$ into the contiguous sequence of events that it represents, is sequentially consistent. 

\begin{example}\label{eg:sc-for-block-programs}
Consider the program in Figure~\ref{fig-onewriter-example}. Let us call it $\program$. Then if we choose the set of preemption points $\points_1$ as in Example~\ref{eg:block-program}, $\sigma^{\points_1}=b^1_1\cdot b^3_1 \cdot b_1^2\cdot b_2^1\cdot b_2^2 $ is an interleaving for the block-program $\program^{\points_1}$. Note that this is not sequentially consistent as the last write on $x$ before $\rd(x,2)$ in $\sigma^{\points_1}$ is $\wt(x,1)$, which conflicts with $\rd(x,2)$. 
 Hence $\rd(x,2)$ has no write to read from.  
An example of a sequentially consistent interleaving of blocks from $\program^{\points_2}$ is: $\sigma^{\points_2}=b_1^1\cdot b_1^3\cdot b_2^1\cdot b_1^2\cdot b_3^1$.
\end{example}

Thus,  there is a direct correspondence between interleavings of $\program$ and interleavings of a block-program $\program^\points$. When consecutive blocks within a thread do not appear together in an interleaving, then, it represents an interleaving of $\program$ where the preemptions occur exactly at $\points$.  We state this observation in the following lemma, whose proof follows directly from the definition of block-programs.

\begin{lemma}
Let $\points$ be a set of preemption points of $\program$. Then, $\program$ has a sequentially consistent interleaving where the preemptions occur at $\points$, iff the block-program $\program^{\points}$ has a sequentially consistent interleaving where no two consecutive blocks come from the same thread. 
\end{lemma}

The question therefore boils down to checking whether there is a sequentially consistent interleaving of $\program^{\points}$.

\noindent{\bf{Finding a sequentially consistent interleaving of a block-program}}.  
Recall that each thread has been divided into $c_i + 1$ blocks. A naive solution would be to enumerate all possible orderings of the $\sum_i (c_i + 1)$ blocks satisfying the program order and verify sequential consistency of each of them. However:
\begin{equation} 
  \sum_i (c_i + 1) = \sum_i c_i + \sum_i 1 \le \pp + k  \label{eq:naive} 
\end{equation} 
This only gives a complexity bounded by $(\pp + k)!$ which is not polynomial-time  when $\pp$ is the only fixed parameter.
Therefore, we propose a different procedure. 

\paragraph*{Dividing the block program into inner blocks and outer blocks.} Suppose $\program^{\points}$ is a block-program defined on the set of preemption points $\points$. We call the last block of any thread in $\program^{\points}$ as outer block and all other non-last blocks as inner blocks. In Figure~\ref{fig:example-algo}, all the violet blocks are inner blocks while the green blocks are outer blocks. As each inner block ends with a preemption point in $\points$, there are total $\pp$ inner blocks. Then the total number of outer blocks is $k$. We then  do a pre-processing that would add some ordering among the outer blocks of $\program^{\points}$ and discard the chosen set of preemptions $\points$ immediately if a cycle is found.

\subsection{Step 2.2. Conflict graph analysis on outer blocks}

By the end of Step 2.2, we have now identified the inner blocks and the outer blocks of $\program^{\points}$. 

We say the block $b_w$ conflicts block $b_r$ if (1) they belong to different threads, (2) there is a variable $x$ such that last write on $x$ in $b_w$ is in conflict with some read on $x$ in $b_r$. Then we order them as $b_r \to b_w$. Now look at the set of all outer blocks of $\program^{\points}$. For each outer block $b_w$ find if it is in conflict with some other outer block $b_r$. If yes, put the the order of conflict $b_r \to b_w$. This says that in any interleaving, $b_r$ appears before $b_w$. As $b_w$ is the last block of its thread, if $b_w$ writes to $x$ and is in conflict with $b_r$, $b_w$ cannot be executed before $b_r$. The blue edges in Figure~\ref{fig:example-algo} represents conflict order among the outer blocks. 

We can construct a conflict graph $G^{\program
^{\points}}$ as follows: For each outer block there is a node in $G$. For each outer block $b_w$ if it conflicts some other outer block $b_r$, we put the edge $b_r \to b_w$.

If there is a cycle in the conflict graph we declare that there exists no SC interleaving for the chosen set of preemptions $\points$. Thus we proceed to the next step only if the conflict graph is acyclic. 

\paragraph*{Complexity}: For each of the $k$ outer blocks we check if it conflicts some other outer block. To do this, given an outer block $b$, for every variable $x$ that the block writes, we check if the last write on $x$ in $b$ is in conflict with any read of other outer blocks. Thus constructing the graph takes $\mathcal{O}(n\cdot k)$ time. Next we check for a cycle. This takes not more than $\mathcal{O}(k^2)$ time. As $k< n$, Step 2.3 would take $\mathcal{O}(n\cdot k)$ time to run. 

\subsection{Step 2.4. Placing the outer blocks in a permutation of inner blocks}

After the conflict analysis returns no cycle for the conflict graph $G^{\program
^{\points}}$, in Step 2.3 we choose a permutation $p$ of the inner blocks. On each of these $\pp!$ permutations we now call Algorithm~\ref{pseudocode:vscp-interleaving}. This is the final step which returns an SC interleaving of the blocks of $\program^{\points}$. For this section $\program$ refers to $\program^{\points}$ and $G$ refers to $G{\program
^{\points}}$. We omit the superscript of the notations.

\begin{algorithm}[h]
\caption{Module to obtain an $\scmm$-consistent block interleaving from a given block program $\program$}\label{pseudocode:vscp-interleaving}

\Input{A block-program $\program$ of a $\onewriter$ program, a permutation $\sigma=\{b_1,b_2,\cdots,b_{\pp}\}$  on inner blocks and an acyclic conflict graph $G$ on the outer blocks}
\Output{an $\scmm$ consistent interleaving of blocks}
\Fn{CheckSC($\program$, $\sigma$, $G$)}{

    $G'=$ a linearization of $G$;
    
    $test=0$;
    
    $b'=\epsilon$;
    
    Add dummy block $b_{\pp+1}$ to $\sigma$;
    
    \For{$i\gets0$ \KwTo $\pp+1$}{\label{for-loop-outer}
        
        \While{there exists an outer block which is enabled}{\label{While-loop}
        
        Choose the earliest enabled block $b$ in $G'$;
        
        \For{$j\gets i+1$ \KwTo $\pp$}{\label{For-loop-iner}
                \If{$b$ conflicts $b_j$} {\label{if-conflict}
                
                test$=1$; 
                
                \Break;}
            
            }
            
            \If{$test=0$}{
            
            Place $b$ between $b'$ and $b_{i+1}$ in $\sigma$;
            
            $b'=b$;
            
            Remove $b$ from $G'$;}
        
        }
        
        $b'=b_{i+1}$;

    }
    
    \If{$\sigma$ is $\scmm$-consistent and $G'$ is empty}{\Return 1;}
}

\end{algorithm}

The pseudocode describes Step 2.4 of the polynomial time algorithm for \vscp. The input is a block-program of a single writer system  divided into inner and outer blocks, a permutation of the inner blocks $\sigma$ and a conflict graph $G$ on the outer blocks. Recall that outer blocks are the last blocks of each block sequence. It takes a linearization $G'$ of $G$ and uses it to place the outer blocks in their correct position in $\sigma$ so that the resulting interleaving is $\scmm$-consistent. 

Description: The algorithm~\ref{pseudocode:vscp-interleaving} chooses a linearization $G'$ of $G$. For each position between the consecutive inner block in $\sigma$, beginning from the left to $b_1$, the algorithm places all enabled outer blocks such that the conflict order of $G$ is satisfied. Note that any outer  block $b$ is enabled at some position only if (1) all blocks above it in its thread are already in sigma before the currently considered position (2) if the last block before the current position is $b'$ and the prefix of $\sigma$ upto $b'$ is $\sigma'$, then $\sigma'\cdot b$ should be $\scmm$-consistent. The variable $b'$ captures the last position in $\sigma$ where a block has been inserted. Moreover, $b$ should not conflict any inner block of $\sigma$ at any later position. This is captured by the condition of the \emph{If} statement in line~\ref{if-conflict}. At the end of the outer for-loop, if no more outer blocks remain to be inserted in $\sigma$, and the resulting interleaving $\sigma$ is $\scmm$-consistent, return 1. 

Complexity: Finding a linearization of $G$ requires $O(k^2)$ time. Note that the for-loop in line~\ref{for-loop-outer} runs $\mathcal{O}(\pp)$ times while testing the while loop condition takes $\mathcal{O}(n)$ time units. The inner for-loop takes no more than $\mathcal{O}(n)$ time to run. Any other statements inside the while-loop does not take more than $\mathcal{O}(n)$ time. The while-loop can atmost run $k$ times as there are atmost $k$ outer  threads. Thus the total complexity of the procedure is  $\mathcal{O}(\pp nk+k^2)=\mathcal{O}(\pp nk)$ as $k<n$.

\begin{theorem}
 A block program $\program$ has an $\scmm$-consistent interleaving of all its blocks if and only if the  algorithm~\ref{pseudocode:vscp-interleaving} returns 1 for some permutation of the inner blocks of $\program$.
\end{theorem}

\begin{proof}
$(\Leftarrow)$ This direction is direct. Whenever the algorithm returns 1 for some permutation of the inner blocks of $\program$, there is an $\scmm$-consistent execution of $\program$. This is because we insert any outer block at some position if and only if it does not conflict any inner block later to that position and it is enabled there. Also any conflict between the outer blocks is respected while the placing the outer blocks. This is supported by the fact that only the earliest enabled block of $G'$, a linearization of the conflict graph $G$, is chosen at each iteration of the while-loop in algorithm~\ref{pseudocode:vscp-interleaving}. Also the last check for $\scmm$-consistency strengthens our claim.

 $(\Rightarrow)$ Let $\program$ has an $\scmm$ consistent interleaving $\sigma$  for some permutation of its inner block $p$. We shall show that with $p$ as an input, the algorithm will return 1. Note that there must be no conflicting cycle among the outer blocks or else $\prog$ wouldn't have any $\scmm$-consistent execution. Let us define a few auxiliary notions. 
 
 Let $P=\{b_1,b_2,\cdots,b_{\pp}\}$ be a permutations of the $\pp$ inner blocks of $\program$. $S_P$ be the set of all $\scmm$-consistent interleavings of $\program$ with order of inner blocks $p$. 
 Let $\sigma_1,\sigma_2\in S_P$. We say $\sigma^b$ is a prefix of $\sigma$ upto block $b$. Then we say $\sigma_1<\sigma_2$ if $\sigma_1^{b}$ has equal or more number of blocks than $\sigma_2^{b}$ for all blocks $b$ in $p$. Then there exists an interleaving $\sigma'$ such that $\sigma'<\sigma$ for all $\sigma\in S_P$. We call such an interleaving a `nice' interleaving. We shall show that a nice interleaving can be generated by our algorithm. 

 Let $R$ be the set of all interleavings obtained from our algorithm on $p$. 
Let $b$ be the first block of $\sigma$. If $b$ is an inner block, the prefix of $\sigma$ upto $b$ matches all interleavings of $\sigma$. Else if $b$ is an outer block, $b$ must be enabled at the first iteration of the while-loop in algorithm~\ref{pseudocode:vscp-interleaving}. As $b$ is an outer block, it must not conflict any $inner$ block of $p$ or else it would have been not placed as the first block in $\sigma$. Note that this could be said only because we are dealing with $\onewriter$. For multiple writers there would be more outer blocks that can conflict some inner blocks on the same variable as $b$. Thus $b$'s place would not be determined. 

Let $\sigma^b$ be the longest prefix that matches some interleaving in $R$. Surely, Let us look at $\sigma^b.b'$. If $b'$ is an inner block, there must not be any enabled outer block that does not conflict with the inner blocks not included in $\sigma^b$. Otherwise, $\sigma^b$ would not be `nice'. There would be an interleaving $\sigma'$ in $R$ such that $\sigma'<\sigma$. But then there must be an interleaving in $R$ with prefix $\sigma^b.b'$. 

On the other hand if $b'$ is an outer block, (1) $\sigma^b.b'$ must be $\scmm$-consistent (2) $b'$ cannot conflict any block not in $\sigma^b$. Our algorithm then is free to choose $b'$ after $\sigma^b$. Hence $\sigma^b.b'$ must be the prefix of some interleaving in $R$. By induction, our claim follows.

 \end{proof}

%%%%%%%%%%%%%%%%%
%%%%%%%%%%%%%%%%%
%%%%%%%%%%%%%%%%%

\section{Appendix for Section~\ref{sec:hardness}}
\label{app:hardness}

\threeWriterLemma*

\begin{proof}
    ($\Rightarrow$)
    Suppose $A$ is a satisfying assignment for $\varphi$. We construct a partial interleaving $\sigma$ as follows. Note that since we want $\sigma$ to have $0$ preemptions, each thread appears as a contiguous block.

    First, for each variable $x_i$, add the events of the thread $S^{A(x_i)}$, in some arbitrary order over $i$. 
    Next, for each $i$, add the thread $S_{x_i}$ if $A(x_i)=1$, and include $S_{\neg x_i}$ if $A(x_i)=0$. 

    Since $A$ satisfies $\varphi$, for every clause $C_j$, at least one of the above threads writes $\wt(c_j,1)$ before any read of $c_j$ occurs (see Table~\ref{tab:reduction-threads-scmm}).
    Now add the thread $S_f$ to the interleaving. 
    By construction, for each read $\rd(c_j,1)$ in $S_f$, there is a preceding write $\wt(c_j,1)$ in $\sigma$ with no intervening conflicting write. Therefore, these reads respect the $\scmm$ semantics. The final write $\wt(x,1)$ in $S_f$ ensures that subsequent reads on $x$ read the value $1$.

    Next, for each $i$, add the helper thread $H_i^{A(x_i)}$. 
    Each such thread contains a read $\rd(x,1)$, which reads the value written by the write $\wt(x,1)$ in $S_f$, and a read $\rd(v_i, A(x_i))$, which reads the value written by  the earlier write in $S_i^{A(x_i)}$.
    Finally, add the remaining threads: for each $i$, include $S_i^{\neg A(x_i)}$, $H_i^{\neg A(x_i)}$, and the literal thread corresponding to the negation of $A(x_i)$. All reads in these threads reads the value written by earlier writes in $\sigma$.

    By construction, every read in $\sigma$ has the value written by the most recent preceding write to the same variable with the same value. Hence, $\sigma$ is sequentially consistent. Further, since each thread appears as a contiguous block, $\sigma$ has $0$ preemptions.

    ($\Leftarrow$) Suppose there exists a sequentially consistent partial interleaving $\sigma$ of $\program_\varphi$ with $0$ preemptions.
    Since $\sigma$ has $0$ preemptions, $\sigma$ is essentially just a concatenation of the threads in some order. 

    Each helper thread $H_i^0$ and $H_i^1$ contains a read $\rd(x,1)$, and hence must appear after the write $\wt(x,1)$ in $S_f$.

    Consider a variable $x_i$. If both $S_i^0$ and $S_i^1$ appear before $S_f$ in $\sigma$, then both $\wt(v_i,0)$ and $\wt(v_i,1)$ occur before $S_f$. Since no further writes to $v_i$ occur before the helper threads, only one of the reads $\rd(v_i,0)$ or $\rd(v_i,1)$ can have a preceding matching write with no intervening conflicting write. Thus, at most one of $H_i^0$ or $H_i^1$ can appear after $S_f$, contradicting completeness of $\sigma$. Hence, for each $i$, at most one of $S_i^0$ or $S_i^1$ appears before $S_f$.

From the construction of $\sigma$, we obtain the following assignment
\[
A(x_i) =
\begin{cases}
0 & \text{if } S_i^0 \text{ appears before } S_f \text{ in } \sigma,\\
1 & \text{otherwise.}
\end{cases}
\]
Now consider any clause $C_j$. The read $\rd(c_j,1)$ in $S_f$ requires a preceding write $\wt(c_j,1)$ in $\sigma$ with no intervening conflicting write. Such a write can only come from a thread corresponding to a literal of $C_j$ scheduled before $S_f$. By definition of $A$, this literal evaluates to true. Thus, each clause of $\varphi$ is satisfied.
\end{proof}

\twoWriterLemma*

\begin{proof}
    The proof follows the same structure as Lemma~\ref{lem:eth-sc}.

    ($\Rightarrow$) Suppose that $A$ is a satisfying assignment for $\varphi$. We construct a partial interleaving $\sigma$ with $0$ preemptions.
    For each variable $x_i$, add the thread $S_i^{A(x_i)}$ and the corresponding literal thread ($S_{x_i}$ if $A(x_i)=1$, otherwise $S_{\neg x_i}$).

    Now, consider a clause $C_j$. If the third literal $\ell_j^3$ is true under $A$, include the thread corresponding to $\ell_j^3$ before $S_f$. Then the write $\wt(d_j,1)$ appears before $S_f$, and the read $\rd(d_j,1)$ in $S_f$ is satisfied.
    Otherwise, one of $\ell_j^1$ or $\ell_j^2$ is true. 
    Include the corresponding literal thread (which writes $\wt(c_j,1)$), followed by $S_j$. Then the read $\rd(c_j,1)$ in $S_j$ has a preceding write with no intervening conflicting write, and $\wt(d_j,1)$ in $S_j$ appears before $S_f$.
    After handling all clauses, add the thread $S_f$. Finally, add all remaining threads. By construction, each read has a preceding matching write with no intervening conflicting write, and $\sigma$ has $0$ preemptions.

\medskip

    ($\Leftarrow$) Suppose there exists a sequentially consistent partial interleaving $\sigma$ with $0$ preemptions. Then $\sigma$ is a concatenation of threads. As in Lemma~\ref{lem:eth-sc}, for each variable $x_i$, at most one of $S_i^0$ or $S_i^1$ appears before $S_f$. Define
\[
A(x_i) =
\begin{cases}
0 & \text{if } S_{i}^0 \text{ appears before } S_f,\\
1 & \text{otherwise.}
\end{cases}
\]

    Consider a clause $C_j$. The read $\rd(d_j,1)$ in $S_f$ must have a preceding write $\wt(d_j,1)$ in $\sigma$ with no intervening conflicting writes. This write can either copme from the thread corresponding to $\ell_j^3$, or from $S_{j}$. If it comes from the thread for $\ell_j^3$, then that thread appears before $S_f$, and hence $\ell_j^3$ evaluates to true under $A$.
    Otherwise, it comes from $S_{j}$. In this case, the read $\rd(c_j,1)$ in $S_{j}$ must have a preceding write $\wt(c_j,1)$ in a literal thread corresponding to $\ell_j^1$ or $\ell_j^2$ that appears before $S_{j}$, and hence before $S_f$. The corresponding literal is therefore true under $A$.

    Thus, in all cases, at least one literal in $C_j$ is true under $A$. Hence $\varphi$ is satisfiable.\qedhere
\end{proof}

\subsection{Proof of the reduction presented in Section~\ref{sec:w1-hard}}
\label{app:w1-hard}

The correctness of the construction is explained through the next two lemmas.

\begin{lemma}\label{lem:independent-set-to-SC-interleaving}
If $G$ has an independent set of size $k$, there is an SC interleaving for $\program_G$ with $3k$ preemptions.
\end{lemma}
\begin{proof} 
Suppose there is an independent set $\{v_1, v_2, \dots, v_k\}$ of size $k$. Here is an SC interleaving of $\program_G$ with $4k$ preemptions:
\begin{description}
\item[Phase 0.] First execute $\Init$ completely.
\item[Phase 1.] Pick $\Sel_1$. Suppose $v_1 = u_{m}$. Execute $\BB^1_{u}$ upto $\BB^1_{u_{m-1}}$. Notice that so far SC has not been violated. Further, the last write on $y_e$ variables is still $\wt(y_e, 0)$. Now, pick $\Sel_2$ and do the same, and execute all operations before the start of $\BB^2_{v_2}$. Repeat this for each of the selector threads $\Sel_3$ until $\Sel_k$. This phase incurs $k$ preemptions. 

\item[Phase 2.] Execute $\Sel_1$ upto $\wt(x_1, 0)$. At this point the last write on edges incident on $v_1$ is $\wt(y_e, 1)$. Since we started with an independent set, none of these edges are incident on $v_j$ for $j > 1$. So, indeed, we can execute $\Sel_j$ upto $\wt(x_j, 0)$ for each $j \ge 1$ in some order, since no two pair of them have a common edge. This phase ends with $k$ preemptions too.

\item[Phase 3.] Now, we completely execute $\Checker_1$, followed by $\Checker_2$, and so on, until $\Checker_k$. No extra preemptions get added. 
\item[Phase 4.] For each $\Sel_j$ we execute upto the end of the current block $\BB^j_{v_j}$ so that all the variables $y_e$ have executed the last write $\wt(y_e, 0)$ at the end of this phase. This phase needs $k$ more preemptions. Total number of preemptions so far is $3k$.
\item Finally, we can execute all the $\Sel_j$ threads one by one without incurring any preemptions. 
\end{description} 
The total number of preemptions in all the phases is thus $3k$ and by construction, the interleaving is sequentially consistent.
\end{proof}

\begin{lemma}\label{lem:sc-interleaving-to-independent-set}
If $\program_G$ has an SC interleaving, then $G$ has an independent set of size $k$.
\end{lemma}
\begin{proof} 
Suppose there is an SC interleaving of $\program_G$. Here are some observations.
\begin{itemize}
\item Before the $\rd(x_j, 0)$ event of $\Checker_j$, the last write on $x_j$ should be $\wt(x_j, 0)$. Hence $\Sel_j$ should have executed the operations inside some block $\BB^j_{v_j}$ upto $\wt(x_j, 0)$.  

\item Due to the $p_j$ variables, $\Checker_1$ executes completely before $\Checker_2$. Similarly, $\Checker_2$ executes completely before $\Checker_3$ starts, and so on.
\item Notice the $\wt(s, 0)$ in $\Checker_1$, and a $\rd(s, 1)$ in every block $\BB^j_{v_j}$. Since, (1) all instructions upto $\wt(x_j, 0)$ need to be done before $\Checker_j$, and (2) $\rd(s, 1)$ cannot happen after $\Checker_1$, we deduce that all the selectors have executed upto $\rd(s, 1)$ before the $\wt(s,0)$ of $\Checker_1$. In particular, the block $\Start^j_{v_j}$ is executed before $\wt(s, 0)$ of $\Checker_1$.

\end{itemize} 
The crucial point is that: before the $\wt(s,0)$ of $\Checker_1$, we have executed the operation $\wt(y_e, j)$ for all edges incident on $v_j$, due to the block $\Start^j_{v_j}$. We claim that there can be no edge between $v_i$ and $v_j$ when $i \neq j$. Suppose there is an edge $e = (v_i, v_j)$ between two selected vertices. Thread $\Sel_i$ has a $\wt(y_e, i)$ and thread $\Sel_j$ has a $\wt(y_e, j)$ before executing $\wt(x_i, 0)$ and $\wt(x_j, 0)$ respectively. Notice the reads $\rd(y_e, i)$ in $\Sel_i$ and $\rd(y_e, j)$ in $\Sel_j$ in the $\Finish$ blocks, which need to be executed, after the $\Checker$ threads. This would not be possible in an SC interleaving. We conclude that there can be no common edge between $v_i$ and $v_j$. This shows that the selected vertices form an independent set, of size $k$. 
\end{proof}

The two lemmas above prove the correctness of our reduction. If $G$ has an independent set of size $k$, we get an interleaving in $\program_G$ with at most $3k$ preemptions (Lemma~\ref{lem:independent-set-to-SC-interleaving}). If $\program_G$ has an SC interleaving with at most $3k$ preemptions, in particular, this means there is an SC interleaving. Lemma~\ref{lem:sc-interleaving-to-independent-set} shows there is an independent set of size $k$.

\end{document}